\documentclass{article}
\usepackage{lmodern}
\usepackage[T1]{fontenc}
\usepackage[latin9]{inputenc}
\usepackage{amsmath}
\usepackage{amsthm}
\usepackage{amssymb}
\usepackage{graphicx}
\usepackage{color}
\usepackage{authblk}
\usepackage{fancyhdr}
\makeatletter

\providecommand{\tabularnewline}{\\}

\theoremstyle{definition}
\newtheorem{defn}{\protect\definitionname}
\theoremstyle{plain}
\newtheorem{thm}{\protect\theoremname}
\theoremstyle{plain}
\newtheorem{prop}{\protect\propositionname}
\theoremstyle{definition}
\newtheorem{example}{\protect\examplename}
\theoremstyle{definition}
\newtheorem*{defn*}{\protect\definitionname}

\makeatother

\usepackage[english]{babel}
\providecommand{\definitionname}{Definition}
\providecommand{\examplename}{Example}
\providecommand{\propositionname}{Proposition}
\providecommand{\theoremname}{Theorem}

\begin{document}
	
	\title{Individual Differential Privacy: A Utility-Preserving Formulation of 
		Differential Privacy Guarantees}
	\author[1]{Jordi Soria-Comas\thanks{jordi.soria@urv.cat}}
	\author[1]{Josep Domingo-Ferrer\thanks{josep.domingo@urv.cat}}
	\author[1]{David S\'anchez\thanks{david.sanchez@urv.cat}}
	\author[2]{David Meg\'ias\thanks{dmegias@uoc.edu}}
	\affil[1]{UNESCO Chair in Data Privacy,
			Universitat Rovira i Virgili.}
	\affil[2]{Internet Interdisciplinary Institute (IN3), Universitat Oberta de Catalunya.}
	
	
	\date{}
	
	\fbox{\begin{minipage}[t]{.9\columnwidth}%
			\begin{center}
				\Large
				This is a preprint of \url{http://dx.doi.org/10.1109/TIFS.2017.2663337}
				\par\end{center}%
		\end{minipage}}
	
	\begingroup
	\let\newpage\relax
	\maketitle
	\endgroup
	
	\lhead{}
	\rhead{}
	\chead{This is a preprint of
	 \url{http://dx.doi.org/10.1109/TIFS.2017.2663337}}
\begin{abstract}
Differential privacy is a popular privacy model within the research
community because of the strong privacy guarantee it offers, namely
that the presence or absence of any individual in a data set does
not significantly influence the results of analyses on the data set.
However, enforcing this strict guarantee in practice significantly
distorts data and/or limits data uses, thus diminishing the analytical
utility of the differentially private results. In an attempt to address
this shortcoming, several relaxations of differential privacy have
been proposed that trade off privacy guarantees for improved data
utility. In this work, we argue that the standard formalization of
differential privacy is stricter than required by the intuitive privacy
guarantee it seeks. In particular, the standard formalization requires
indistinguishability of results between \emph{any} pair of 
neighbor data
sets, while indistinguishability between \emph{the actual} data set and its
neighbor data sets should be enough. This limits the data controller's
ability to adjust the level of protection to the actual data, hence
resulting in significant accuracy loss. In this respect, we propose
\emph{individual differential privacy}, an alternative differential
privacy notion that offers {\em the same privacy guarantees as standard
differential privacy to individuals} (even though not to groups
of individuals). This new notion allows the data controller to adjust
the distortion to the actual data set, which results in less distortion
and more analytical accuracy. We propose several mechanisms to attain
\emph{individual differential privacy} 
and we compare 
the new notion against standard
differential privacy in terms of the accuracy of the analytical results. \end{abstract}

\section{Introduction}

The development of information technologies has boosted the collection
of data about individuals. Data stores containing personal data are
a valuable resource to support decision making both in the public
and the private sector. However, individuals have growing concerns
on the use of their potentially sensitive data. In order to
guarantee the fundamental right of individuals to privacy, appropriate
data protection measures to limit disclosure risk must be implemented
before making data available for analysis.

Essentially, two main approaches exist in statistical disclosure
control~\cite{hundepool2012,book2016} to limit the disclosure risk:
(i) non-interactive protection, whereby a protected version of the
original data set collected from the data subjects is generated and
released, and (ii) interactive protection, whereby a user-queried
data analysis is performed on the original data set, and then a protected
version of the results is returned to the user. When the type of data
analysis is not known at the time of data protection, the former approach
is the only viable solution; however, for a fixed data analysis known
beforehand, interactive protection should be preferred as it allows
adjusting the level of protection to the analysis being performed,
which makes it possible to maximize the accuracy of the results.

In this work, we focus on the interactive setting, which differential
privacy~\cite{Dwork2006} has substantially advanced. The main contribution
of differential privacy is the strong privacy guarantee it provides:
while it does not prevent disclosure, it guarantees that a disclosure
is equally likely whether or not any particular individual contributes
her data. This is diametrically opposed to the aim of privacy models
focusing on the non-interactive setting~\cite{Samarati2001,Machanavajjhala2007,Li2007,Soria2012kanonprob},
which seek to provide absolute guarantees against specific types of
intruders (with a specific knowledge on the data to be protected).
The problem of the latter models is that real intruders may differ
from the specific intruders 
considered by the models; in particular,
real intruders may know more than assumed.

However, despite its popularity among researchers and the step forward
it offers in terms of privacy guarantees, differential privacy is
only being deployed to a limited extent in real-world applications.
The basic reason is the poor accuracy/utility of differentially private
results. Except for a number of well-behaved applications (queries
that are stable to modification of one record), differential privacy
has too large an impact on utility for it to be widely used in data
analysis/mining (see related remarks in Section~\ref{sec:rel}).
This is also confirmed by the amount of relaxations of differential
privacy that have been proposed (see Section~\ref{sec:rel}), which
seek to improve the accuracy of the results by trading it off against
privacy guarantees.

In any case, the intuition behind differential privacy, that is, \emph{``the
presence or absence of an individual in a data set should not significantly
modify the results of the analysis''}, is very adequate for disclosure
risk limitation: if the data on an individual have a significant impact
on the results of an analysis, most probably the privacy of this individual
is at risk. Thus, intuitively, differential privacy ensures the data
are adequately protected.

However, we argue here that the standard formalization of differential
privacy is stronger (and, thus, leads to more data distortion and 
less useful results) than what the above data protection
intuition requires. This can be seen from two different but related
points of view: 
\begin{itemize}
\item Differential privacy assumes the presence of a trusted party that
holds the data set, receives queries submitted by the users and returns
differentially private results for these queries. In the standard differential
privacy formulation, the trusted party is not allowed to use its knowledge
of the actual data set 
when computing the noise to be added 
to the query response in
order to protect privacy. 
\item Differential privacy provides provable privacy guarantees to groups
of data subjects (in addition to those that derive from the guarantees
given to each data subject individually). These guarantees for groups
go beyond the previously mentioned intuition behind differential
privacy, which is about protecting single individuals. 
\end{itemize}
While being more demanding than required by the intuitive
notion of differential privacy may be meaningful, we argue
in the following sections that enforcing an actually stronger privacy
notion may have a significant negative impact on the accuracy/utility
of the protected results. Thus, if data utility is a priority, the
suitability of including additional guarantees (i.e., group-based
differential privacy) should be carefully pondered.
Besides, in statistical disclosure control, the goal
is to prevent accurate inferences on single individuals, but 
accurate inferences on groups (statistical
and aggregate computations) are usually viewed as legitimate analyses.

In this line, we propose here {\em individual differential privacy},
a novel formalization of differential privacy that is aligned with
the intuitive differential privacy guarantees. Our approach
limits the protection to the individuals in the data set (as suggested
in the intuitive notion), rather than extending it to groups of individuals
(as formalized in the standard definition of differential privacy).
In practice, this means that the trusted party managing the data can
make use of its knowledge of the actual data set at the time of query
response and downwardly adjust the distortion to the actual data;
as a consequence of the lower data distortion, the accuracy/utility of
the protected results can be significantly improved in comparison
with the standard formulation of differential privacy.

Nonetheless, we will also indicate how to extend individual differential
privacy to group differential privacy, in case the data 
controller wants to offer privacy guarantees
for groups while still 
leveraging his knowledge on the actual data set when computing the 
distortion to be added for protection.  

The rest of this paper is organized as follows. Section~\ref{sec:dp}
provides background on the standard formulation and the usual mechanisms
to attain differential privacy. Section~\ref{sec:rel} reviews related
work aimed at improving the accuracy/utility of differentially private
results and highlights the differences with our approach. In Section~\ref{sec:idp},
we present and formalize our \emph{individual differential privacy}
model and in Section~\ref{sec:mechanism} we propose several mechanisms
to satisfy it that exploit its less strict formulation to reduce data
distortion. Section~\ref{sec:Evaluation} analyzes the accuracy of
individually differentially private results of several common queries
and compares their accuracy against the one obtained using standard
differential privacy. The final section presents the conclusions along
with some future research lines.

\section{Background on Differential Privacy\label{sec:dp}}

Differential privacy was originally proposed in~\cite{Dwork2006}
as a privacy model for the interactive setting, that is, to protect
the results of queries to a database. In this setting, a differentially
private sanitization mechanism sits between the user submitting queries
and the database controller answering them. To preserve the privacy of
individuals, the sanitization mechanism must guarantee that the contribution
of each individual's data to a query result is limited (according
to an $\epsilon$ parameter).
\begin{defn}[$\epsilon$-differential privacy]
\label{def:dp}A randomized function $\kappa$ gives $\epsilon$-differential
privacy (or $\epsilon$-DP) if, for all data sets $D_{1}$ and $D_{2}$
that differ in one record (a.k.a. neighbor data sets), and
all $S\subset Range(\kappa)$, we have 
\[
\Pr(\kappa(D_{1})\in S)\text{\ensuremath{\le}}\exp(\epsilon)\Pr(\kappa(D_{2})\in S).
\]
\end{defn}

Let $D$ be the data set that is to be protected. Given a query $f$,
the goal in DP is to find an approximation to $f$ that satisfies
the privacy requirements stated in Definition~\ref{def:dp}. Let
us call $\kappa_{f}$ such an approximation. The value of $\kappa_{f}(D)$
is then returned as the query result, instead of the actual value
$f(D)$.

In the following, we refer to the differentially private condition
in terms of indistinguishability: the distributions of the responses
to the queries between data sets that differ in one record are similar.

An interesting property of DP, which the privacy models in the $k$-anonymity
family ($k$-anonymity~\cite{Samarati2001}, $l$-diversity~\cite{Machanavajjhala2007},
$t$-closeness~\cite{Li2007}, etc.) do not have, is composition:
composing several differentially private results still satisfies
DP, although with a different $\epsilon$ parameter value. Several
composition theorems have been stated. The most basic ones are:
\begin{thm}[Sequential composition]
\label{teo1} Let $\kappa_{1}$ be a randomized function giving $\epsilon_{1}$-DP
and $\kappa_{2}$ a randomized function giving $\epsilon_{2}$-DP.
Then, any deterministic function of $(\kappa_{1},\kappa_{2})$ gives
$(\epsilon_{1}+\epsilon_{2})$-DP. 
\end{thm}

\begin{thm}[Parallel composition]
\label{teo2} Let $\kappa_{1}$ and $\kappa_{2}$ be randomized functions
giving $\epsilon$-DP. If $\kappa_{1}$ and $\kappa_{2}$ are applied
to {\em disjoint} data sets or subsets of records, any deterministic
function of $(\kappa_{1},\kappa_{2})$ gives $\epsilon$-DP. 
\end{thm}

The computational mechanism to attain DP is often called a differentially
private sanitizer. Differentially private mechanisms for numerical
data can be seen as noise addition mechanisms that, rather than
returning the actual query value $f(D)$, mask it by adding some noise.
The amount of noise that needs to be added depends on the variability
of the query function between neighbor data sets. This noise may
be: (i) independent of the data set, which requires adjusting the
noise to the maximum variability between neighbor data sets, or
(ii) dependent on the data set, which allows adjusting the level of
noise to the variability of $f$ in the actual data set.

\subsection{Noise calibration to the global sensitivity\label{sub:laplace}}

The global sensitivity measures the maximum variability of a function
between neighbor data sets.
\begin{defn}[$l_{1}$-sensitivity (a.k.a. global sensitivity)]
\label{def:global_sensitivity}The $l_{1}$-sensitivity of a function
$f:\mathcal{D}\rightarrow\mathbb{R}^{k}$ is 
\[
\Delta f=\max_{\begin{array}{c}
{\scriptstyle {x,y\in\mathcal{D}}}\\
{\scriptstyle {d(x,y)=1}}
\end{array}}\left\Vert f(x)-f(y)\right\Vert _{1},
\]
where $d(x,y)$ means that data sets $x$ and $y$ differ in one record. 
\end{defn}
In calibration to the global sensitivity, DP is attained by adding a noise
to the query response that is proportional to the global sensitivity;
that is, the maximum variability of the query response in the domain
of the data. Several noise distributions are possible (e.g. the Laplace
distribution~\cite{DworkMNS06}, the optimal absolutely continuous
distribution~\cite{Soria13}, and the discrete Laplace distribution~\cite{Inusah2006}).
For the sake of conciseness, we focus on the commonly used Laplace
distribution.

\begin{prop}
Let $f$ be a query function with values in $\mathbb{R}^{k}$. The
mechanism $\kappa_{f}(X)=f(X)+(N_{1},\ldots,N_{k})$, where $N_{i}$
are independent and identically distributed random noises drawn from a
$Laplace(0,\epsilon/\Delta f)$ distribution, 
is $\epsilon$-differentially private. 
\end{prop}

\subsection{Noise calibration to the smooth sensitivity\label{sub:smooth}}

Even if the variability of the query function in the actual data set
is small, quite often its variability in the data domain (i.e., global
sensitivity) is large. In these cases, the use of a noise calibrated
to the global sensitivity may seriously and unjustifiedly damage data,
but a data-dependent noise may still permit accurate results. Given
a data set $D$, the variability of a function between $D$ and its
neighbor data sets is known as the \emph{local sensitivity}.
\begin{defn}[Local sensitivity~\cite{Nissim2007}]
\label{def:local}Let $f$ be a function that is evaluated at data
sets and returns values in $\mathbb{R}^{k}$. The local sensitivity
of $f$ at $D$ is 
\[
LS_{f}(D)=\max_{y:d(y,D)=1}\left\Vert f(y)-f(D)\right\Vert _{1}.
\]

\end{defn}
Obviously, the global sensitivity upper-bounds the local sensitivity.
Moreover, the local sensitivity is usually small, except for especially
ill-conditioned data sets. In this case, the gap between the global
and the local sensitivity can be large. This is illustrated in the
following example for a function that returns the median of a list
of values.
\begin{example}
Consider a data set $X=\{x_{1},\ldots,x_{n}\}$, where each record
corresponds to a value in $\{0,1\}$. For the sake of simplicity,
we assume that the number of records is odd, say $n=2m+1$, so that
the median corresponds to a single record, the $m+1$-th record. The
global sensitivity of the median is 1, since we can consider the neighbor
data sets:

\[
\begin{array}{c}
\{0,\stackrel{m+1}{\ldots},0,1,\stackrel{m}{\ldots},1\}\rightarrow\textit{median}=0,\\
\{0,\stackrel{m}{\ldots},0,1,\stackrel{m+1}{\ldots},1\}\rightarrow\textit{median}=1.
\end{array}
\]

The local sensitivity is, except for the two previous data sets, always
0. The reason is that, except for the previous data sets, changing
the value of a record does not modify the median. 
\end{example}
In most cases, releasing the value of a query with a magnitude of
noise proportional to the local sensitivity (rather than the global
sensitivity) would result in a significantly more accurate response.
However, using the local sensitivity in mechanisms designed for global
sensitivity does not yield DP.
\begin{example}
Consider the data sets: $\{0,\stackrel{m+2}{\ldots},0,1,\stackrel{m-1}{\ldots},1\}$
and $\{0,\stackrel{m+1}{\ldots},0,1,\stackrel{m}{\ldots},1\}$. In
both cases the median is 0, but the local sensitivity differs; it
is 0 in the first case and 1 in the second one:

\begin{tabular}{ccc}
 & \textit{median}  & \textit{local sensitivity}\tabularnewline
$X=\{0,\stackrel{m+2}{\ldots},0,1,\stackrel{m-1}{\ldots},1\}$  & 0  & 0\tabularnewline
$X'=\{0,\stackrel{m+1}{\ldots},0,1,\stackrel{m}{\ldots},1\}$  & 0  & 1
\end{tabular}\vspace{2mm}

Given that the local sensitivity for $X$ is 0, adding a noise proportional
to the local sensitivity does not modify the median; thus, the probability
of getting 1 is 0. For $\epsilon$-DP to be satisfied, the probability
of getting 1 for $X'$ must also be 0. However, that is not the case,
since the local sensitivity for $X'$ is different from 0. 
\end{example}
The previous example shows that the amount of noise used to protect
a data set should not only be proportional to its local sensitivity
but also take into account the local sensitivity of neighbor data
sets. Considering the local sensitivity of neighbor data sets leads
to the notion of \emph{smooth sensitivity}.
\begin{defn}[Smooth sensitivity]
For $\beta>0$ and a data set $D\in\mathcal{D}$, the $\beta$-smooth
sensitivity of $f$ at $D$ is defined as 
\[
S_{f,\beta}(D)=\max_{y\in\mathcal{D}}(LS_{f}(y)\exp(-\beta d(D,y))).
\]

\end{defn}

The greater the parameter $\beta$, the smaller the dependence of
the smooth sensitivity on the local sensitivity of neighbor data
sets. Thus, the amount of noise required to attain $\epsilon$-DP
must depend on other factors apart from the smooth sensitivity. In particular,
a special kind of distributions, known as $(\alpha,\beta)$-admissible
distributions~\cite{Nissim2007},
 must be used. The downside of these distributions is
that they are heavy-tailed.

\section{Related Work\label{sec:rel}}

Even though DP is appreciated for the strong privacy guarantees it
offers, its practical deployment is hampered by the poor accuracy
offered by differentially private mechanisms.

Works that discuss the accuracy limitations of differentially private
results in different contexts include~\cite{Bambauer2013,Sarathy2011,Friedman2010,Mohan12,Fredrikson14}.
Although these accuracy issues have hindered the practical adoption
of DP, they have also boosted further research to find fixes. This
research has taken two main lines: (i) to come up with novel differentially
private mechanisms that improve the accuracy of the results, and (ii)
to propose relaxations of DP that require less data distortion and
allow for more accurate results.

An outcome of the first line is the design of general mechanisms that
improve accuracy with respect to the basic Laplace noise addition
mechanism. In~\cite{Nissim2007}, the calibration of Laplace noise
to the global sensitivity of the data is replaced by the calibration
of suitable noises to the smooth sensitivity. On the other hand, the
authors of~\cite{Soria13} showed that the Laplace distribution
is not optimal to attain DP based on calibration to the global sensitivity.
They described and constructed the optimal absolutely
continuous distributions: essentially,
a distribution is optimal if the probability mass is as concentrated
as possible around zero given the DP constraints. Other outcomes of
the first line of research are mechanisms that are less sensitive.
In~\cite{Soria2013,Soria2014,Sanchez2016}, several methods based on microaggregation
of records are proposed to generate differentially private data sets.
In~\cite{Zhang2014}, the dependence between attributes is analyzed
to reduce the dimensionality in the computation of differentially
private histograms. Other works that try to improve the accuracy in
histogram publication are~\cite{Xu2012,Xiao2010}. In~\cite{Jagannathan2012},
a differentially private alternative to the ID3 algorithm for learning
decision trees is proposed.

The second line of research, focused on finding relaxations of DP,
has been active since the inception of DP. In~\cite{DworkKMMN06},
the concept of $\delta$-approximate $\epsilon$-indistinguishability
is presented (a.k.a. $(\epsilon,\delta)$-indistinguishability in~\cite{Nissim2007}).
This relaxation allows some additional margin $\delta$ to the requirements
in DP. In~\cite{Mach08}, the notion of $(\epsilon,\delta)$-probabilistic
differential privacy (a.k.a. $(\epsilon,\delta)$-pdp) is proposed.
Rather than allowing some additional margin, it requires $\epsilon$-DP
to be satisfied with probability greater than $1-\delta$. In other
words, the probability that the adversary gains significant information
about an individual is, at most, $\delta$. Yet another relaxation
is given in~\cite{Mohan12}, who assume that confidential data become
less sensitive over time, which allows relaxing privacy parameters
for older data. In~\cite{Mach15}, a relaxation is presented that
restricts the definition of neighbor data sets: they are no longer
data sets differing in any record, but in a record within a certain subset.
In~\cite{DworkR16}, an alternative relaxation of DP, called $(\mu,\tau)$-concentrated
differential privacy, is proposed. Similarly to $(\epsilon,\delta)$-pdp,
concentrated differential privacy allows the ratio of probabilities
to be arbitrarily large with a small probability that is determined
by the parameters $\mu$ and $\tau$.

In all the above relaxations, the accuracy gain is obtained by allowing
the differentially private condition to be broken: the presence or
absence of an individual may leak some information, although not too
much or only with a small probability. The relaxation that we propose
in the next section is different in the sense that the privacy guarantees
for individuals envisioned in the original definition of DP are preserved.

\section{Individual differential privacy\label{sec:idp}}

The existing relaxations of DP violate strict DP because, under some
circumstances, they permit the probability of responses to differ
significantly between neighbor data sets. In this section, we take
a different approach to improve data accuracy/utility. Rather than
proposing another relaxation, we analyze the intuitive privacy guarantees
that DP seeks to capture and we suggest an alternative definition
that allows exactly attaining these privacy guarantees and nothing
more.

\subsection{Intuitive view of differential privacy\label{sub:intuitive}}

As introduced above, DP assumes the presence of a trusted party that:
(i) holds the data set, (ii) receives the queries submitted by the
data users, and (iii) responds to them in a privacy-aware manner.

We described in Section~\ref{sec:dp} that the standard notion of
DP is formalized in terms of indistinguishability of the response
to queries between neighbor data sets. However, prior to this formalization,
\cite{Dwork2006} provided an intuitive description of the privacy
guarantees aimed at by DP: 
\begin{quotation}
Any given disclosure will be, within a small multiplicative factor,
just as likely whether or not the individual participates in the database.
As a consequence, there is a nominally higher risk for the individual
who participates, and only a nominal gain for the individual who conceals
or misrepresents her data. 
\end{quotation}
Thus, individuals should not be reluctant to participate in the data set.
After all, the risk of disclosure is only very marginally increased 
by participation.

Providing relative (rather than absolute) privacy guarantees is 
the only sensible approach to deal with intruders having 
arbitrary side knowledge.
Absolute privacy guarantees against arbitrary side knowledge 
are incompatible with releasing accurate statistics. 
This is illustrated in~\cite{Dwork2006} with a simple example about the
height of an individual $I$: if the intruder knows 
that the height of $I$ is two
inches above the average height of the population of a country, then 
releasing an accurate approximation of the average height of the population 
provides the intruder with an accurate estimate of $I$'s height. 
Here it is important to notice that such
disclosure happens even if $I$ does not participate in the data set.
Thus, the disclosure is not associated to $I$'s participation in the
data set, but to the accurate side knowledge. This illustrates
that differential
privacy cannot guarantee that disclosure will not happen; what it 
guarantees is that, should disclosure happen, it will not be due
to the participation of any particular individual.

\subsection{Formalization of individual differential privacy}

DP is the result of formalizing the intuitive view of privacy described
in Section~\ref{sub:intuitive} into a rigorous mathematical definition.
However, the fact that in this formalization (see Definition~\ref{def:dp})
results are required to be indistinguishable between \emph{any} pair of
neighbor data sets is more stringent than required by the intuitive
view of privacy that is described above.
Let us explain this in greater detail.

Consider an individual $I$ who has to decide between participating
in a data set or not. To neutralize any reluctance by $I$ to
disclose her private information, $I$ is told 
that query answers based on the data
set will not allow anyone to learn anything that was 
not learnable without $I$'s 
presence; this is precisely the intuitive privacy guarantee DP offers. 
Although DP controls the access to a
data set regardless of the way the data were collected (either
through voluntary participation or not), thinking in terms of
voluntary participation is clarifying: we can think that an 
individual accepts to participate if she regards
privacy protection as sufficient. 

To attain such privacy guarantees, DP requires the response to be 
indistinguishable between any pair of neighbor data sets. 
While such a requirement yields the target privacy guarantees indeed,
it is an overkill because, in reality, queries are not answered 
on an arbitrary data set, but on the actual data set held by the
trusted party. In other words, if $D$ is the collected data set,
the target privacy guarantees can be attained by  
just requiring indistinguishability
of the responses between $D$ and its neighbor data sets.
Notice that, although the data
set $D$ is not known until all the individuals have made their 
decisions about participating/contributing to it,
the presence of a trusted party that controls the access to $D$ makes
it possible to give the previously described indistinguishability
guarantees. After all, the data set $D$ is known to the trusted party 
at the time of query response. 

According to the previous discussion, we propose the following privacy
model, which we call $\epsilon$-individual differential privacy.

\begin{defn}[$\epsilon$-individual differential privacy]
	\label{def:fixed_instance}
	Given a data set $D$, a response mechanism
	$\kappa(\cdot)$ satisfies $\epsilon$-individual differential
	privacy (or $\epsilon$-iDP) if, for any data set
	$D'$ that is a neighbor of $D$, and 
        any $S\subset Range(\kappa)$ we have 
	\[
	\begin{array}{l}
	\exp(-\epsilon)\Pr(\kappa(D')\in S)\le\Pr(\kappa(D)\in S)\\
	\text{\ensuremath{\le}}\exp(\epsilon)\Pr(\kappa(D')\in S).
	\end{array}
	\]
\end{defn}

In line with Definition~\ref{def:dp}, we require the probability
of any result to differ between neighbor data sets at most by
a factor of $\exp(\epsilon)$. However, unlike in Definition~\ref{def:dp}, 
the role of the data sets $D$ and $D'$ is not exchangeable:
$D$ refers to the actual data set, and $D'$ to a neighbor data
set of $D$. 
The asymmetry between $D$ and $D'$ is relevant, because
indistinguishability is achieved only between $D$ and its 
neighbor data sets.
As a side effect of this asymmetry, we need to explicitly
enforce an upper bound ($\Pr(\kappa(D)\in S)\text{\ensuremath{\le}}\exp(\epsilon)\Pr(\kappa(D')\in S)$)
and a lower bound ($\exp(-\epsilon)\Pr(\kappa(D')\in S))\le\Pr(\kappa(D)\in S)$).
This was not needed in Definition~\ref{def:dp} because the upper
bound could be obtained from the lower bound by exchanging the roles
of $D$ and $D'$.

To illustrate the implications of the differences between DP and iDP,
consider a data set $D$ whose records are either 0 or 1, and assume
that we query the data set for the median. We consider two cases, 
$D_{1}=\{0,0,0,0,1\}$ and $D_{2}=\{0,0,0,1,1\}$, which are
discussed next:
\begin{itemize}
\item In data set $D_{1}$, the median is $0$. Moreover, the median remains
$0$ after modification of any single record. Since the median is unaffected
by the contribution of any single individual, there is no way to learn
anything about 
an individual from the median of $D_1$. Hence, according to 
iDP, the median of $D_{1}$ can be released without any masking. This
is not the case according to DP, where the mere existence of a 
pair of neighbor
data sets (not involving $D_1$) 
with different values for the median requires masking it
(even if the median of $D_1$ discloses nothing about any 
individual). 
\item In data set $D_{2}$, the median is $0$, but a change in a single
record is enough for the median to become $1$. 
Since a single individual
may affect the result of the median, information about an individual
might be learned from the value of the median. To limit the risk
of disclosure, both DP and iDP require masking the outcome of the median.
\end{itemize}

In the previous example, we have considered boundary cases where the
change of a single record can make it necessary or unnecessary to mask the
query result. Let us now consider a different case, consisting of 
a data set $D_3$ with 99 binary records, of which 90 contain the
value 1. Clearly, no single individual
has any effect on the median of $D_3$. Thus, iDP does 
not require any masking for the answer to the median query. 
In contrast, under DP the possible
existence of
 a data set $D_4$ containing forty-nine 0s and fifty 1s has to be 
considered, in which a single individual can change the median; 
since this data set must be  
indistinguishable from its neighbors given the value of its median,
the answer to the median query must be masked even for $D_3$
(although $D_3$ and $D_4$ are far from neighbors).

Like in DP, two different notions of neighbor data sets $D$ and $D'$
 are possible in iDP. We can either think of $D'$ as a data set generated
from $D$ by modifying one record or as a data set
generated from $D$ by adding or removing one record. Although such
a choice can make some difference in terms of query answers, 
 both alternatives 
are similar as far as disclosure risk limitation is concerned. 
In the sequel, we assume that $D'$ is generated from $D$ by modifying one
record.

Strictly speaking, \emph{individual differential privacy} is a relaxation
of DP where indistinguishability is only required between the actual
data set $D$ and its neighbor data sets. Thus, the following holds: 
\begin{prop}
	Any mechanism that satisfies $\epsilon$-DP also satisfies $\epsilon$-iDP
	for any actual data set $D$. 
\end{prop}
However, {\em unlike other relaxations in the literature, individual differential
	privacy preserves the strict privacy guarantees that DP gives to 
	individuals}.
In this sense, rather than being regarded as a relaxation, individual
differential privacy can be construed as a more precise formalization
of the intuitive notion of privacy described in Section~\ref{sub:intuitive}.

Even though we can satisfy $\epsilon$-iDP by resorting to the mechanisms
commonly used to satisfy $\epsilon$-DP, doing so would squander the
potential accuracy gains of $\epsilon$-iDP. The reason is that, as
discussed in Section~\ref{sec:dp}, being able to adjust the noise
to the actual data set may significantly reduce its magnitude. In
Section~\ref{sec:mechanism}, we design mechanisms that are specific
to $\epsilon$-iDP.

\subsection{Disclosure limitation for groups of data subjects}

If $\epsilon$-iDP offers all the intuitive privacy guarantees
of $\epsilon$-differential privacy 
while allowing more accuracy, it is because
$\epsilon$-iDP offers {\em nothing more} than those intuitive
guarantees. Let us discuss this issue here.

Even if DP seeks to protect privacy by limiting the impact of \emph{each
single} individual on a query response, Definition~\ref{def:dp}
also results in (limited) privacy guarantees for \emph{groups} of
individuals. For example, if data sets $D_{1}$ and $D_{3}$ differ
in two individuals (two records), by considering the intermediate
data set $D_{2}$ that differs in one individual from $D_{1}$ and
in one individual from $D_{3}$, we have 
\[
\begin{array}{c}
\Pr(\kappa(D_{1})\in S)\le\exp(\epsilon)\Pr(\kappa(D_{2})\in S),\\
\Pr(\kappa(D_{2})\in S)\le\exp(\epsilon)\Pr(\kappa(D_{3})\in S),
\end{array}
\]
which results in 
\[
\Pr(\kappa(D_{1})\in S)\le\exp(2\epsilon)\Pr(\kappa(D_{3})\in S).
\]

That is, $\epsilon$-DP guarantees that the knowledge gain for groups
of two individuals is limited by the factor $\exp(2\epsilon)$. More
generally, the knowledge gain for a group of $n$ individuals is limited
by the factor $\exp(n\epsilon)$.

In contrast, with iDP, the fact that the response mechanism is adjusted
to each concrete data set makes the concatenation of inequalities
not possible. Specifically, for the data sets $D_{1}$, $D_{2}$ and
$D_{3}$ above, iDP provides the upper bounds
\[
\begin{array}{c}
\Pr(\kappa(D_{1})\in S)\le\exp(\epsilon)\Pr(\kappa(D_{2})\in S),\\
\Pr(\kappa'(D_{2})\in S)\le\exp(\epsilon)\Pr(\kappa'(D_{3})\in S),
\end{array}
\]
where $\kappa$ is a mechanism whose distortion is calibrated on 
$D_1$ and $\kappa'$
is a mechanism whose distortion is calibrated on $D_2$.
Note that the probability on the right-hand side of the first inequality
does not match the probability on the left-hand side of the second
inequality.

Therefore, with iDP we lose any {\em direct} privacy guarantees
for groups of individuals. In other words, any disclosure risk limitation
for a group that may subsist under $\epsilon$-iDP is an indirect
consequence of the disclosure risk limitation obtained by each individual
in the group.

Direct group privacy guarantees not being part of the intuitive notion
of DP (which is aimed at individuals), but rather a by-product of
the formalization in Definition~\ref{def:dp}, we drop such group
guarantees in order to improve data utility.
However, the proposed iDP could be modified to 
retain privacy guarantees for
groups of individuals, by requiring indistinguishability between the
data set $D$ and the data sets $D'$ that differ from $D$ in a group
of records, as follows. 

\begin{defn}[$(\epsilon_1,\ldots,\epsilon_n)$-group differential 
privacy]
\label{def:gDP}
Given a data set $D$, a response mechanism $\kappa$ satisfies
 $(\epsilon_1,\ldots,\epsilon_n)$-group
differential privacy (denoted $(\epsilon_1,\ldots,\epsilon_n)$-gDP)
if, for any data set $D'$ at distance $i \in \{1,\ldots,n\}$
from $D$, and any $S\subset Range(\kappa)$
it holds
\[
     \begin{array}{l}
        \exp(-\epsilon_i)\Pr(\kappa(D')\in S)\le\Pr(\kappa(D)\in S)\\
        \text{\ensuremath{\le}}\exp(\epsilon_i)\Pr(\kappa(D')\in S).
        \end{array}
\]
\end{defn}

Notice that, in general, the level of protection for groups
should be smaller than for individuals. For example, in DP, the $\epsilon$
parameter grows linearly with the cardinality of the group and, thus,
the level of indistinguishability decreases exponentially with the
size of the group. As a result, even in standard DP, 
protection for groups
of individuals is only noticeable for small values of $\epsilon$ 
(which on the other hand severely damage data utility) 
and for small groups. 
To attain an equivalent level of protection with
gDP, 
it must be $\epsilon_i=i\epsilon$ for $i \in \mathbb{N}$.

\subsection{Composition theorems}

The composition theorems for DP described in Section~\ref{sec:dp}
are also trivially satisfied by individual differential privacy.
\begin{thm}[Sequential composition]
Let $\kappa_{1}$ be a mechanism giving $\epsilon_{1}$-iDP and
$\kappa_{2}$ a mechanism giving $\epsilon_{2}$-iDP. Then, any
deterministic function of $(\kappa_{1},\kappa_{2})$ gives $\epsilon_{1}+\epsilon_{2}$-iDP. 
\end{thm}

\begin{thm}[Parallel composition]
Let $\kappa_{1}$ and $\kappa_{2}$ be mechanisms giving $\epsilon$-iDP.
If $\kappa_{1}$ and $\kappa_{2}$ are applied to {\em disjoint}
data sets or subsets of records, any deterministic function of $(\kappa_{1},\kappa_{2})$
gives $\epsilon$-iDP. 
\end{thm}
The proofs of the theorems for iDP are not detailed because they are
straightforward adaptations of the proofs of Theorems~\ref{teo1}
and~\ref{teo2} for DP.

\subsection{Privacy axioms}

In~\cite{kifer}, an axiomatization of the notions of privacy 
and utility is proposed. In particular,
two desirable properties of privacy models are described: the axiom 
of transformation
invariance and the privacy axiom of choice.

The axiom of transformation invariance says that 
post-processing of sanitized data
must be safe as long as no sensitive information is incorporated 
into the post-processing.

\begin{defn}[Axiom of transformation invariance] 
	Suppose we
	have a privacy definition, a privacy mechanism $M$ that satisfies
	this definition, and a randomized algorithm $A$ whose
	input space is the output space of $M$ and whose randomness
	is independent of both the data and the randomness in
	$M$. Then $M' = A \circ M$ must also be a privacy mechanism
	satisfying that privacy definition.
\end{defn}

\begin{prop}
	$\epsilon$-iDP satisfies the axiom of transformation invariance.
\end{prop}
\begin{proof}
	Let $M$ be an $\epsilon$-iDP mechanism, and let $A$ be a randomized algorithm
	whose input space is the output space of $M$ and whose randomness
	is independent of both the data and the randomness in
	$M$. According to~\cite{kifer}, we can think of $M$ and $A$ as returning 
	independent probability distributions. That is, $M:\mathcal{D} \rightarrow \mbox{Dist}(E)$  
	assigns to each data set $D$ a probability distribution over $E$ such that, for each $D'$ that differs in
	one record from $D$, we have $\exp(-\epsilon)\Pr(M(D')\in S)$ $\le \Pr(M(D)\in S)$ $\le \exp(\epsilon)\Pr(M(D')\in S)$.
	On the other hand, $A:E\rightarrow \mbox{Dist}(F)$ assigns to each $e\in E$ a probability distribution over $F$.

	To check that $\exp(-\epsilon)\Pr(A(M(D')\in S))$ $\le \Pr(A(M(D))\in S)$ $\le \exp(\epsilon)\Pr(A(M(D'))\in S)$, we decompose $\Pr(A(M(d))\in S)$ as $\int_{x\in S}\Pr(M(d)=y)\Pr(A(y)=x)$ and we use that $M$ satisfies $\epsilon$-iDP.
\end{proof}

The privacy axiom of choice allows randomly choosing between $\epsilon$-iDP mechanisms.
\begin{defn}[Privacy axiom of choice] 
	Let $M_1$ and $M_2$ be privacy mechanisms that satisfy
	a certain privacy definition. For any $p\in [0, 1]$, let $M_p$ be
	a randomized algorithm that, on input $i$, outputs $M_1(i)$ with
	probability $p$ (independent of the data and the randomness
	in $M_1$ and $M_2$) and $M_2(i)$ with probability $1-p$.
	 Then $M_p$
	is a privacy mechanism that satisfies the privacy definition.
\end{defn}

\begin{prop}
	$\epsilon$-iDP satisfies the privacy axiom of choice.
\end{prop}
\begin{proof}
	Let $M_1$ and $M_2$ be $\epsilon$-iDP mechanisms, and let $p\in [0, 1]$.
	Let $M_p$ be a mechanism that, on input $i$, outputs $M_1(i)$ with
	probability $p$ and $M_2(i)$ with probability $1 - p$.
	
	We want to check that $M_p$ is $\epsilon$-iDP. That is, for any $D'$
	neighbor of $D$, we have 
	$\exp(-\epsilon)\Pr(M_p(D')\in S)$ $\le \Pr(M_p(D)\in S)$ $\le \exp(\epsilon)\Pr(M_p(D')\in S)$.
	This is easily seen by expressing $\Pr(M_p(D)\in S)$ as 
	$p\times \Pr(M_1(D)\in S) + (1-p)\times\Pr(M_2(D)\in S)$
	and using that both $M_1$ and $M_2$ satisfy $\epsilon$-iDP.
\end{proof}

\section{$\epsilon$-iDP for Numerical Queries\label{sec:mechanism}}

In DP, response mechanisms for numerical queries are often expressed
in terms of noise addition (see Section~\ref{sec:dp}). In this line,
given a numerical query $f$, this section presents an $\epsilon$-iDP
mechanism of the form $\kappa(x)=f(x)+N$, where $N$ is a random
noise.

For $\kappa(x)=f(x)+N$ to be $\epsilon$-iDP, we have to
make sure that the result of the mechanism is indistinguishable between
$D$ and its neighbor data sets. As we show later, this can be
done by adjusting the noise to the \emph{local sensitivity} $LS_{f}(D)$
(see Definition~\ref{def:local}). The ability to calibrate the noise
to the local sensitivity rather than to the global or smooth sensitivities
is a significant advantage to preserve data utility: 
first, many functions have small 
local sensitivity but large global sensitivity, as noted
right after Definition~\ref{def:local} above; second,
calibration to the local sensitivity also improves the accuracy of query
responses with respect to calibration to the smooth sensitivity,
not only because  
the local sensitivity is smaller than the smooth sensitivity,
but also because calibration to
the local sensitivity allows using exponentially decreasing noise
distributions (rather than the heavy-tailed distributions required
by calibration to the smooth sensitivity). Last but not least, 
calibration to the local sensitivity is simpler than calibration
to the smooth sensitivity.

We next detail how to enforce iDP for both continuous and discrete
numerical results.

\subsection{Laplace mechanism\label{sub:lap}}

The $Laplace(\mu,b)$ 
distribution (a.k.a the double exponential distribution)
is an absolutely continuous distribution whose density function
is 
\[
L{}_{\mu,b}(x)=\frac{1}{2b}\exp\left(-\frac{|x-\mu|}{b}\right),
\]
where $\mu$ is the location parameter, and $b$ is the scale parameter.
We use $L_{b}(x)$ when the mean parameter is zero. 
 
In case of using the Laplace distribution to generate the random noise,
iDP can be rewritten as follows.

\begin{prop}
\label{prop:iDP_Lap}Let $f$ be a query function that takes values
in $\mathbb{R}^{k}$. The mechanism $\kappa(x)=f(x)+(N_{1},\ldots,N_{k})$,
where $N_{i}$ are independent identically distributed 
$Laplace(0,LS_{f}(D)/\epsilon)$
random noises, gives $\epsilon$-iDP. \end{prop}
\begin{proof}
If $N$ is a vector of independent $Laplace(0,\beta)$ distributed
variables, we know that for all $s\in\mathbb{R}^{k}$
\[
\exp\left(-\frac{\left\Vert z-z'\right\Vert _{1}}{\beta}\right)\le\frac{\Pr(z+N=s)}{\Pr(z'+N=s)}\le\exp\left(\frac{\left\Vert z-z'\right\Vert _{1}}{\beta}\right).
\]

Let $D'$ be a data set that differs from $D$ in one record. By taking
$z=f(D)$, $z'=f(D')$, and $\beta=LS_{f}(D)/\epsilon$ we get

\[
\exp(-\epsilon)\le\frac{\Pr(f(D)+N=s)}{\Pr(f(D')+N=s)}\le\exp(\epsilon).
\]

\end{proof}

\subsection{Discrete Laplace mechanism\label{sub:discrete}}

The previous mechanism (based on adding random noise with
values in $\mathbb{R}$) is also capable of providing DP to query
functions with values in $\mathbb{Z}$. However, for such query functions
the use of a noise distribution with support over $\mathbb{Z}$ is
a better option. The discrete version of the Laplace distribution
is defined as follows.
\begin{defn}[Discrete Laplace distribution~\cite{Inusah2006}]
A random variable $N$ follows the discrete Laplace distribution
with parameter $\alpha\in(0,1)$, denoted by $DL(\alpha)$, if for
all $k\in\mathbb{Z}$, 
\[
\Pr(N=i)=\frac{1-\alpha}{1+\alpha}\alpha^{|i|}.
\]

\end{defn}
The discrete Laplace distribution can also be used to attain $\epsilon$-iDP.
To this end, the parameter $\alpha$ must be adjusted to the desired
privacy level $\epsilon$ and to the local sensitivity $LS_{f}(D)$
of the query.
\begin{thm}[The discrete Laplace mechanism]
Let $f$ be a function with values in $\mathbb{Z}^{k}$. The discrete
mechanism $\kappa_{D}(x)=f(x)+N$, where $N=(N_{1},\ldots,N_{k})$
and $N_{i}\sim DL(\exp(-\epsilon/LS_{f}(D)))$ are independent random
variables, gives
 $\epsilon$-iDP.\end{thm}
\begin{proof}
If $N$ is a vector of independent $DL(\alpha)$ distributed variables,
we know that for all $s\in\mathbb{Z}^{k}$
\[
\alpha^{\left\Vert z-z'\right\Vert _{1}}\le\frac{\Pr(z+N=s)}{\Pr(z'+N=s)}\le\alpha^{-\left\Vert z-z'\right\Vert _{1}}.
\]

Let $D'$ be a data set that differs from $D$ in one record. By taking
$z=f(D)$, $z'=f(D')$ and  
$\alpha=\exp(-\epsilon/LS_{f}(D))$,
we get

\[
\exp(-\epsilon)\le\frac{\Pr(f(D)+N=s)}{\Pr(f(D')+N=s)}\le\exp(\epsilon).
\]

\end{proof}

\section{Evaluation\label{sec:Evaluation}}

In this section, we compare the accuracy obtained with the standard
DP definition and with iDP 
for basic noise addition mechanisms. In particular, we compare 
calibration to the global sensitivity (for DP), to the smooth sensitivity (for DP), 
and to the local sensitivity (for iDP).
Noise addition mechanisms provide simple ways to make 
an output compliant with DP (or iDP) and
are agnostic to the specific data uses; 
hence, using standard noise
addition mechanisms is relevant to make a general comparison 
between DP and iDP. 

Since the accuracy we obtain from standard DP calibrated to the
smooth sensitivity and from iDP calibrated to the local sensitivity depends
on the actual data set, a purely empirical evaluation carried out
on some sample data sets would only provide a partial picture. 
Therefore, we give both an empirical analysis based on simulations
and 
a theoretical analysis.
We will see that, 
even in the worst case,
iDP offers a significant improvement on the accuracy of the responses.

\subsection{Median}

The median is a measure of central tendency. Compared to the arithmetic
mean, the median is quite stable, because it is relatively insensitive
to outliers. Since statistical queries like the median represent properties
of the data set rather than properties of a specific individual, we
could expect DP to provide accurate results for such queries. However,
the fact that standard DP guarantees must hold for {\em
any} pair of neighbor data sets forces noise addition
mechanisms to significantly degrade accuracy.
Instead, iDP allows adjusting the amount
of protection to each specific data set; thus, 
most of the times,
iDP yields reasonably high accuracy, as we discuss below.

\subsubsection{Differential privacy via calibration to the global sensitivity}

Let the values $x_{1},\cdots,x_{n}$, with $n=2m+1$, be instances
of an attribute with domain $Dom(A)$. Consider data sets $D_{1}$
and $D_{2}$ such that: in $D_{1}$ values $x_{1},\ldots,x_{m}$ are
equal to the minimum of $Dom(A)$, and values $x_{m+1},\ldots,x_{n}$
are equal to the maximum of $Dom(A)$; and in $D_{2}$ values $x_{1},\ldots,x_{m-1}$
are equal to the minimum of $Dom(A)$, and values $x_{m},\ldots,x_{n}$
are equal to the maximum of $Dom(A)$. As a result, the median $x_{m}$
in $D_{1}$ is the minimum of $Dom(A)$, and the median $x_{m}$ in
$D_{2}$ is the maximum of $Dom(A)$; thus, the global sensitivity
for the median is 
\[
\Delta(\textit{median})=\max(\textit{Dom}(A))-\min(\textit{Dom}(A)).
\]
Having a global sensitivity equal to the domain size severely
compromises the accuracy of differentially private estimations of
the median via calibration to the global sensitivity. The situation becomes
even worse when the domain of the attribute is not naturally bounded
(e.g. incomes are unbounded, even if the income of most individuals
falls within a given window); in such cases, it may not even be possible
to compute the global sensitivity without artificially restricting
the values to a fixed interval.

\subsubsection{Differential privacy via calibration to the smooth sensitivity}

Calibration to the smooth sensitivity tries to avoid the shortcomings
of global sensitivity by adjusting the amount of noise to each data
set. From~\cite{Nissim2007}, the smooth sensitivity for the median
is 
\[
\begin{array}{l}
S_{\textit{median},\epsilon}(D)=\\
{\displaystyle =\max_{k=0,\ldots,n}\{\exp(-k\epsilon)\max_{0\le t\le k+1}\{x_{m+t}-x_{m+t-k-1}\}\}.}
\end{array}
\]
Calibration to the smooth sensitivity is effective at reducing the amount
of noise that is added to most data sets. However, it still has some
drawbacks: 
\begin{itemize}
\item Computing the smooth sensitivity can be complex. In particular, the
formula for the smooth sensitivity of the median has time complexity
$\mathcal{O}(n^{2})$, which may not be feasible for large data sets. 
\item Similarly to the global sensitivity, the smooth sensitivity can only
be computed when $Dom(A)$ is bounded. 
\item The mechanisms used to attain DP are more complex than with global
sensitivity, and the eligible noise distributions are heavy-tailed
(rather than the exponentially decreasing distributions used in calibration
to the global sensitivity). 
\end{itemize}

\subsubsection{iDP (via calibration to the local sensitivity)}

iDP also allows the response mechanism to be independently adjusted
to each data set. iDP can be reached via calibration to the local sensitivity
(see Section~\ref{sec:mechanism}). For the median, the local sensitivity
is: 
\[
LS_{\textit{median}}(D)=\max\{x_{m}-x_{m-1},x_{m+1}-x_{m}\}.
\]
iDP via calibration to the local sensitivity offers the following advantages
compared to calibration to the smooth sensitivity: 
\begin{itemize}
\item The computation of the local sensitivity is less complex. 
\item The local sensitivity is lower than the smooth sensitivity. 
\item The local sensitivity does not depend on the size of the domain. It
only depends on $x_{m-1}$, $x_{m}$ and $x_{m+1}$. 
\item The mechanism used to attain iDP is simpler and the noise distributions
can be selected to be exponentially decreasing (e.g. the Laplace distribution). 
\end{itemize}

\subsubsection{Experimental comparison between DP with smooth 
sensitivity and iDP}

After the previous theoretical comparison, we now turn to experimental
evaluation of the accuracy for the median. Due to the poor accuracy
provided by calibration to the global sensitivity, the comparison focuses
on DP via calibration to the smooth sensitivity and iDP via calibration
to the local sensitivity.

For iDP, we have used the mechanism proposed in Proposition~\ref{prop:iDP_Lap}.
For standard DP, we have used 
a noise $N$ having a density of the form 
\[
f_{N}(x)\propto\frac{1}{1+|x|^{\gamma}},
\]
where $\gamma>1$. 
This noise distribution is $(\frac{\epsilon}{4\gamma},\frac{\epsilon}{\gamma})$-admissible to achieve DP calibrated to the smooth sensitivity~\cite{Nissim2007}.

In the comparison, we take $\epsilon=1$
and we assume that 
the local sensitivity equals the smooth sensitivity.
Since i) noise in DP is proportional to the 
smooth sensitivity, ii) noise in iDP is proportional
to the local sensitivity, and iii) the latter sensitivity is 
never greater than the former,  
both sensitivities being equal is the worst case for iDP.
Specifically, we take both sensitivities equal to 1. 
Figure~\ref{fig:comparison} shows the distribution
of the noise in the response. Even in the most favorable case for
the smooth sensitivity, the accuracy obtained with iDP is much better.
The reason is that, to enforce $\epsilon$-DP via calibration to the smooth
sensitivity, the noise distribution must be heavy-tailed. Table~\ref{tab:conf}
shows 95\%-confidence intervals for these noise distributions.

\begin{figure}
\centering{}\includegraphics[width=8.5cm]{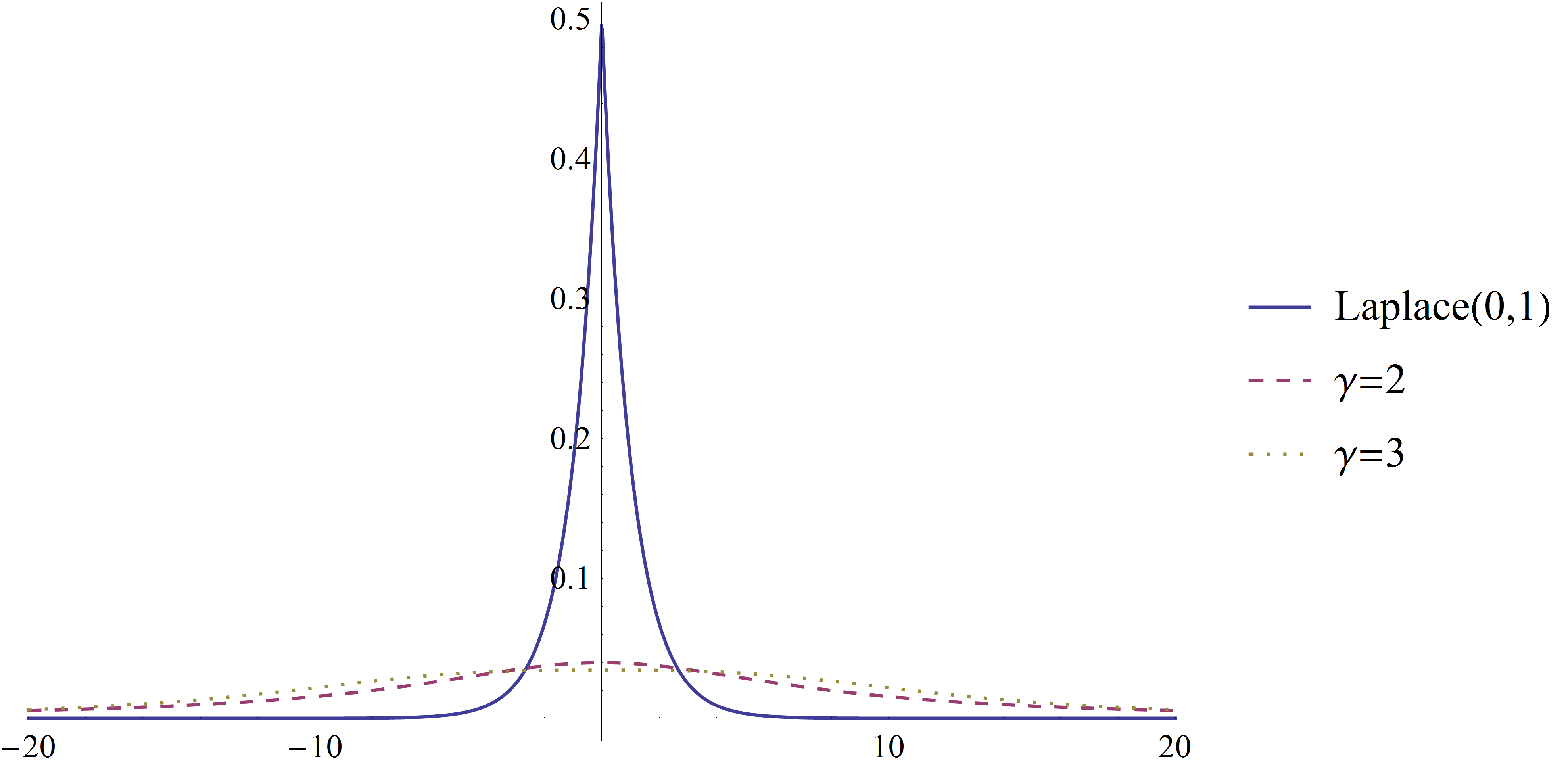}\caption{\label{fig:comparison}Comparison of the noise distributions required
by $\epsilon$-DP via calibration to the smooth sensitivity 
(for $\gamma=2$ and $\gamma=3$, where $\gamma$ is the noise 
parameter) and $\epsilon$-iDP via
calibration to the local sensitivity (with a Laplace distribution), where
the query is the median, $\epsilon=1$ and sensitivities are 1. }
\end{figure}

\begin{table}
\caption{\label{tab:conf}95\% confidence intervals of the noise distributions
required by $\epsilon$-DP via calibration to the smooth
sensitivity 
(for different values of the noise parameter $\gamma$)
and $\epsilon$-iDP via calibration to the local sensitivity (with
a Laplace distribution), where the query is the median, 
$\epsilon=1$ and sensitivities are 1.}
\centering{}%
\begin{tabular}{cccc}
\hline 
\emph{Privacy model}  & $\epsilon$-iDP  & $\epsilon$-DP $\gamma=2$ & $\epsilon$-DP $\gamma=3$\tabularnewline
\hline 
\emph{95\% interval}  & {[}-3,3{]}  & {[}-101.7,101.7{]}  & {[}-34.2,34.2{]}\tabularnewline
\hline 
\end{tabular}
\end{table}

These results show that the use of calibration to the smooth sensitivity
may be incompatible with an accurate and $\epsilon$-differentially
private approximation for the median. Although the smooth sensitivity
is definitely better than the global sensitivity for noise calibration,
at 95\% confidence the amount of noise can be as large as 101 times
the smooth sensitivity. Whether this is acceptable depends on the
data set. Only if the smooth sensitivity is very small compared to
the attribute domain, accuracy may still be acceptable. In contrast,
the accuracy provided by iDP is much better: at 95\% confidence, the
noise is at most 3 times the local sensitivity.

Finally, we evaluate the accuracy of the median by 
computing, for a collection of data sets, 
the average absolute error incurred; we use the absolute 
error because the relative error may be misleading when
the median is close to 0. 
We consider data 
sets with three different sizes  
(10, 100, and 1000) whose records are drawn from $U[0,1]$ (the uniform distribution in $[0,1]$), $N(0,1)$ (the normal distribution with mean 0 
and variance 1) and 
$\mathit{Exp}(1)$ (the exponential distribution with rate 1). 

For the smooth sensitivity to be finite, the domain of the records
must be bounded. Thus, to be able to compute the smooth sensitivity in the 
$N(0,1)$ and the $\mathit{Exp}(1)$ cases, we have to 
restrict the domain of the data.
This is done by bounding the domain to the range between the minimum and 
the maximum values in the data set. 
Figure~\ref{fig:comp} 
shows the average absolute error for each of the data sets. We include 
the results obtained
for DP by adjusting the noise to the smooth sensitivity with $\gamma=3$, and for
iDP by adjusting a Laplace noise to the local sensitivity. 
We took $\gamma=3$ as the best performing value 
after evaluating several other values of $\gamma$;
specifically, smaller values led to too much probability mass 
in the tails, whereas greater values
introduced greater correction factors to the noise distribution 
that offset the
gain produced by having the probability mass closer to zero.
By looking at the figure, we notice that, regardless of the 
underlying data distribution and the data set size, 
the results are significantly better for iDP. 

\begin{figure*}
	\begin{centering}
		\includegraphics[width=3.5cm]{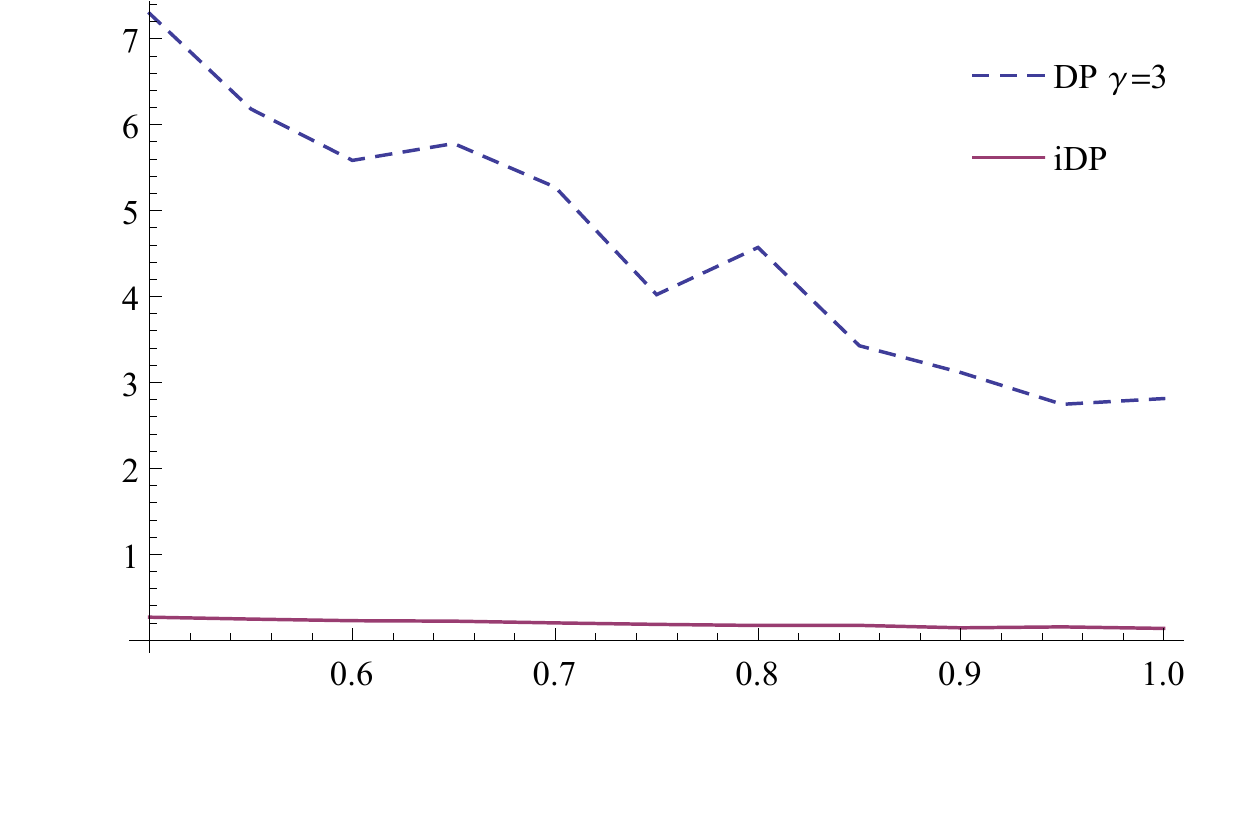}
		\includegraphics[width=3.5cm]{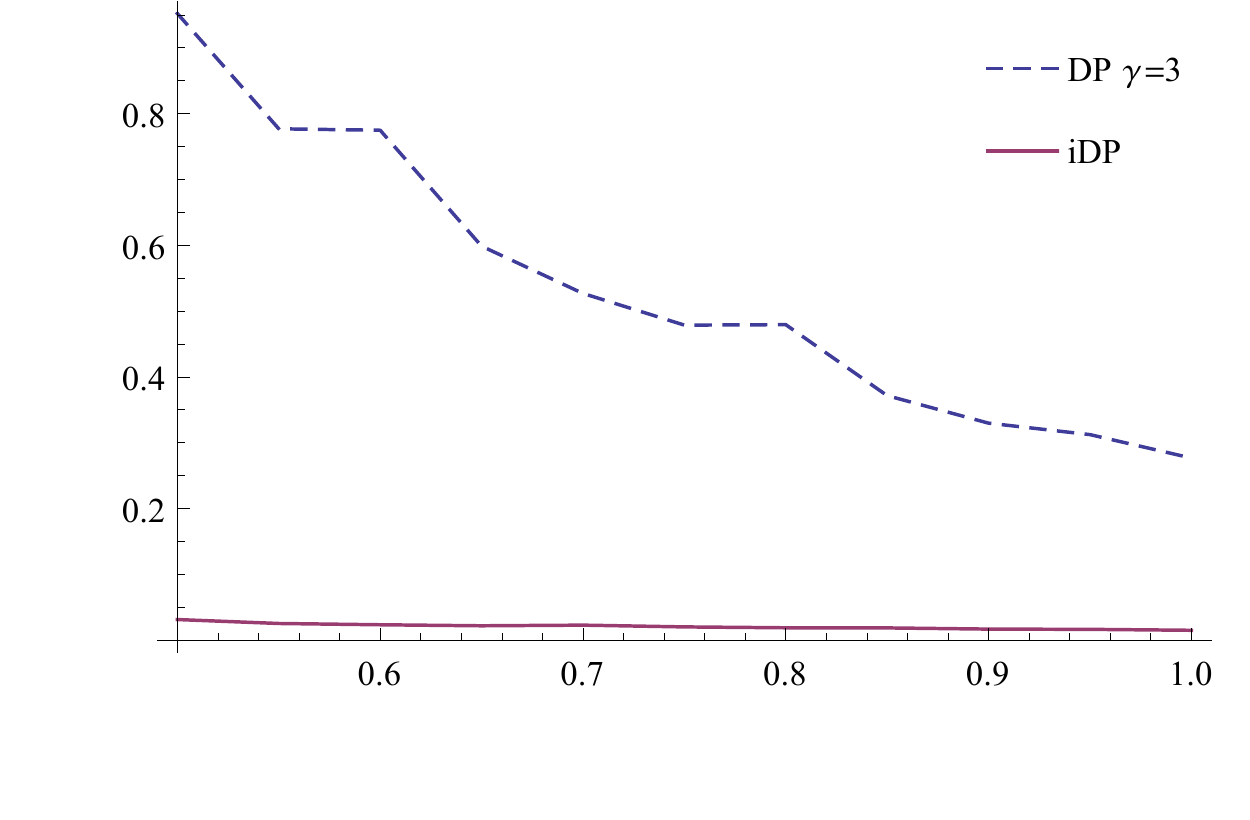}
		\includegraphics[width=3.5cm]{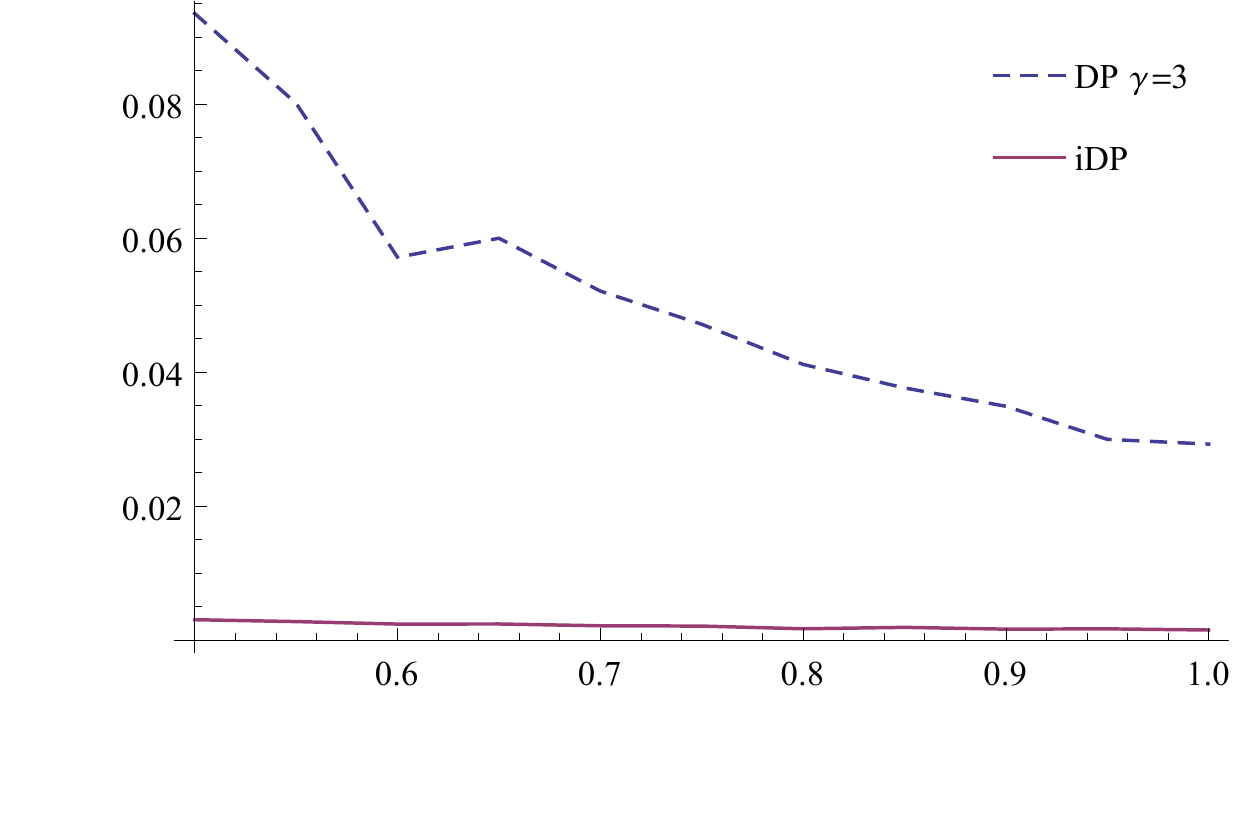}
		\par\end{centering}
	
	\begin{centering}
		\includegraphics[width=3.5cm]{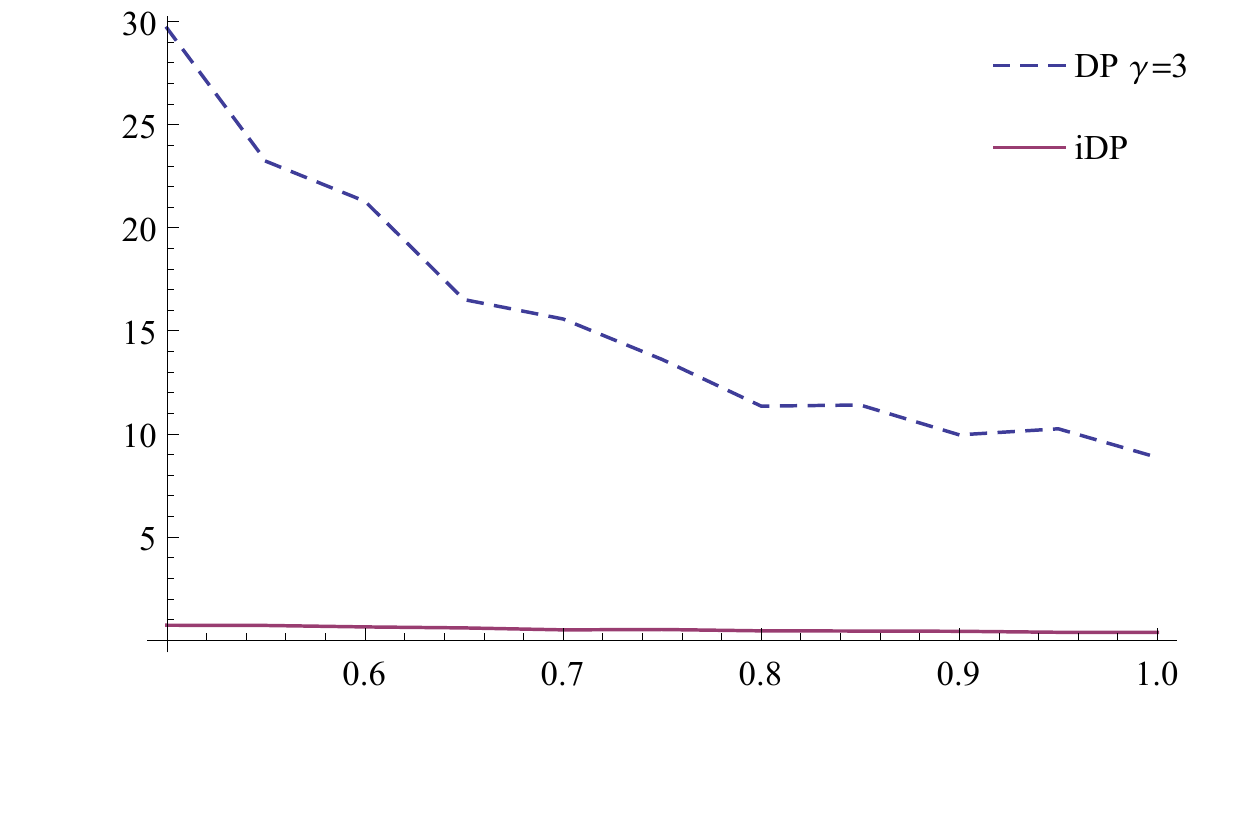}
		\includegraphics[width=3.5cm]{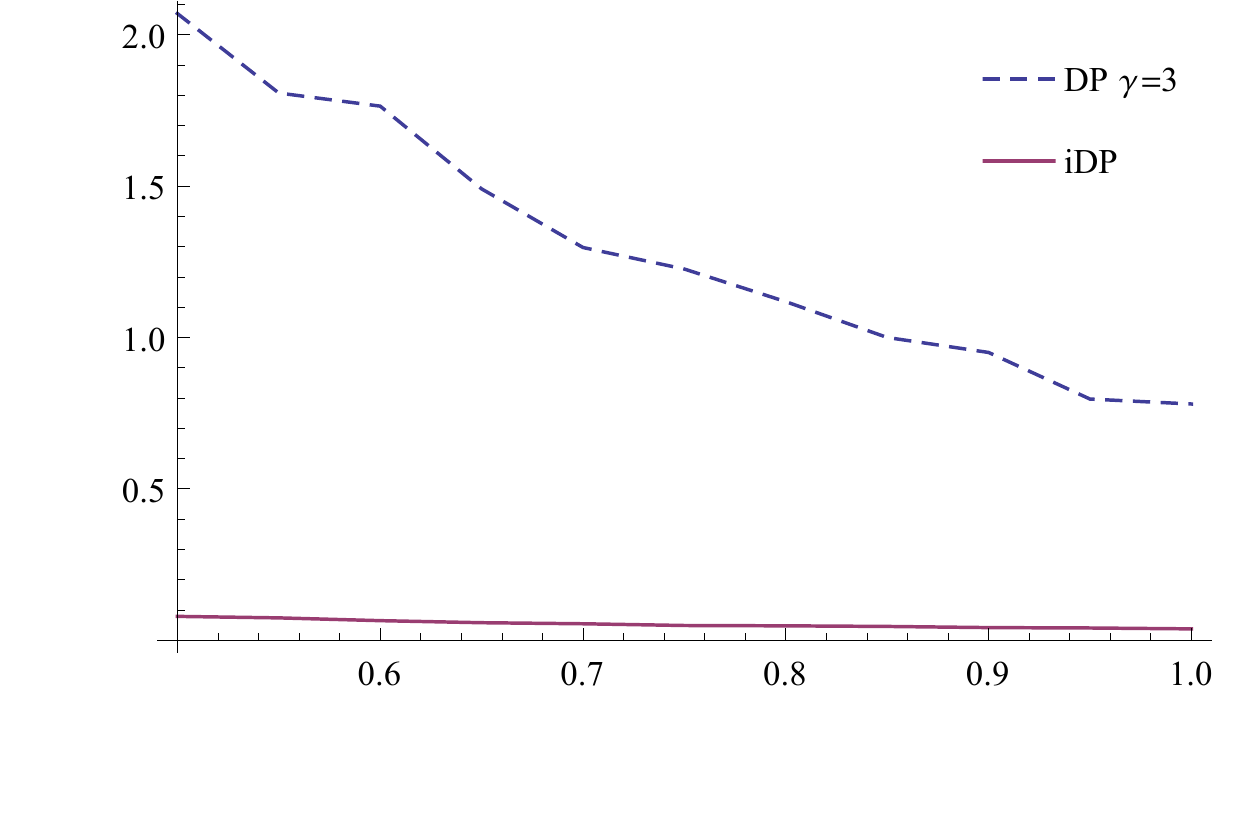}
		\includegraphics[width=3.5cm]{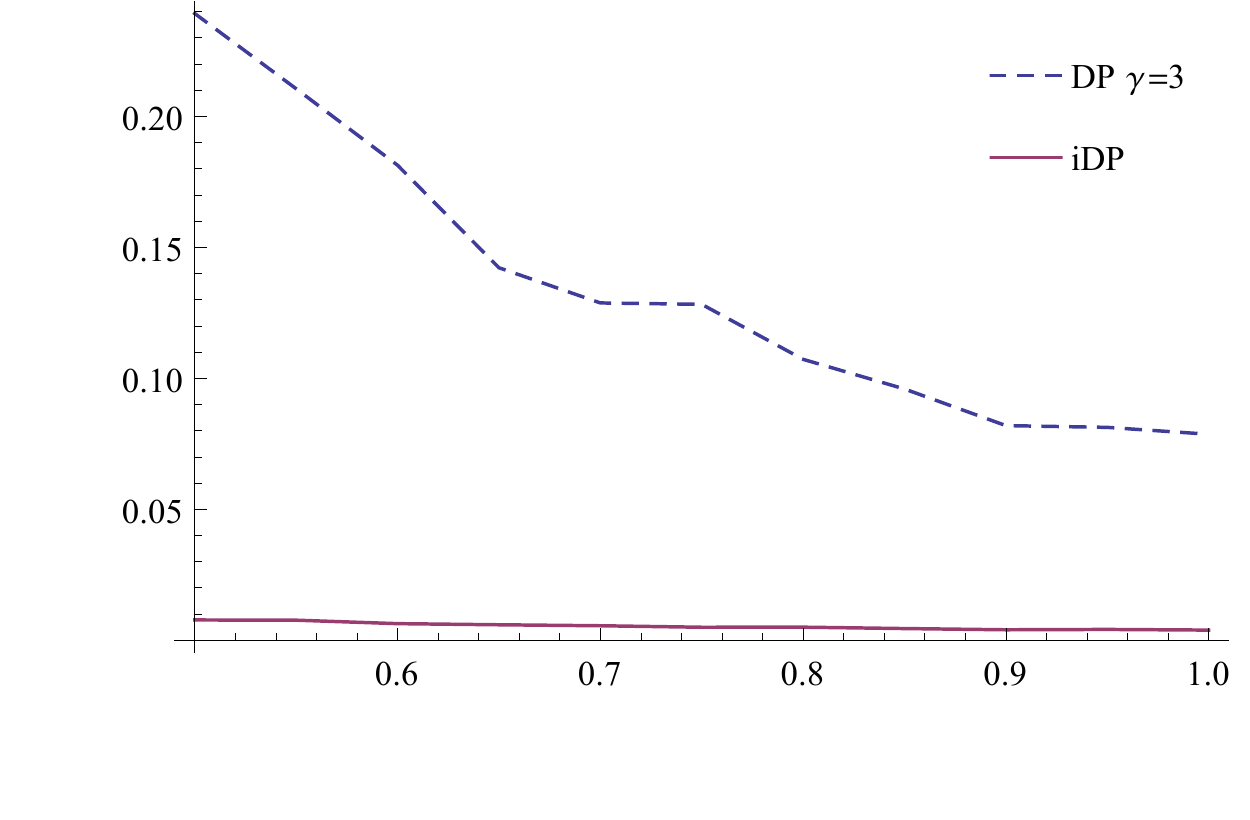}
		\par\end{centering}
	
	\begin{centering}
		\includegraphics[width=3.5cm]{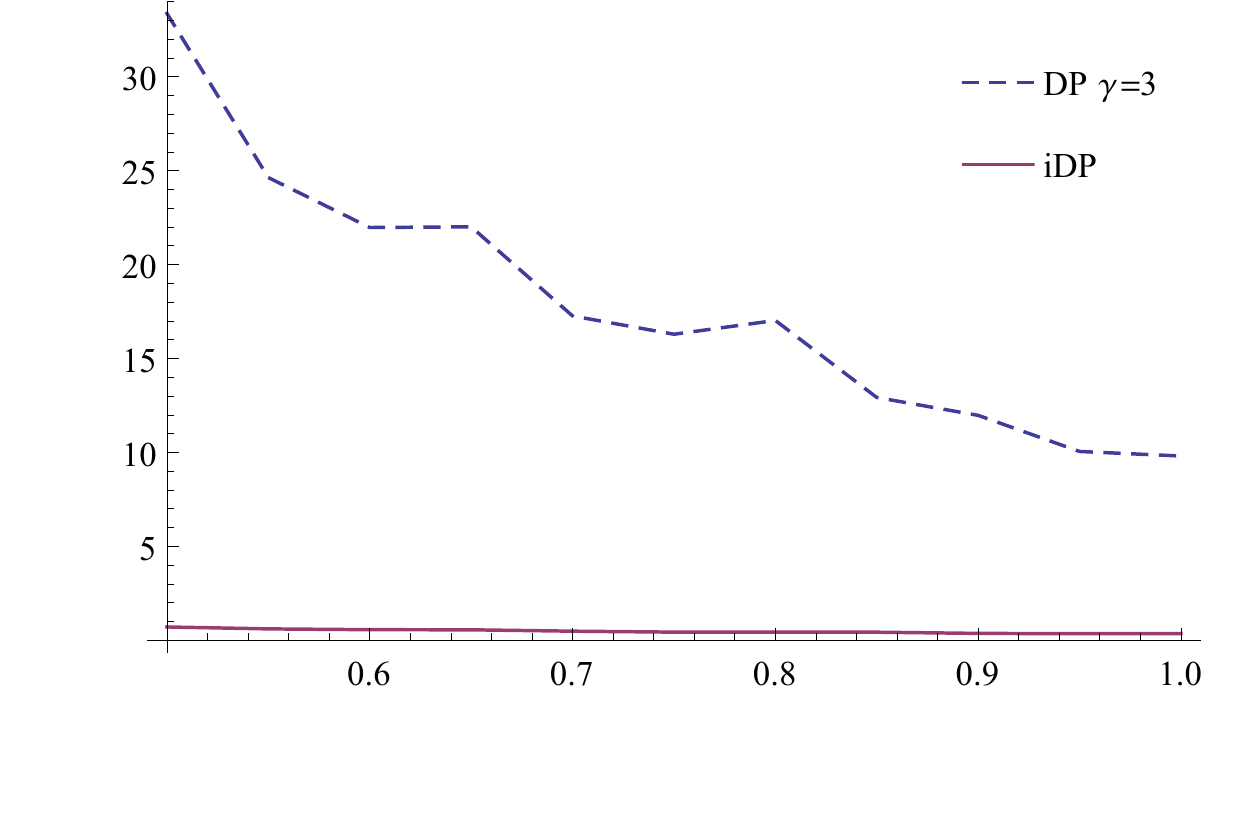}
		\includegraphics[width=3.5cm]{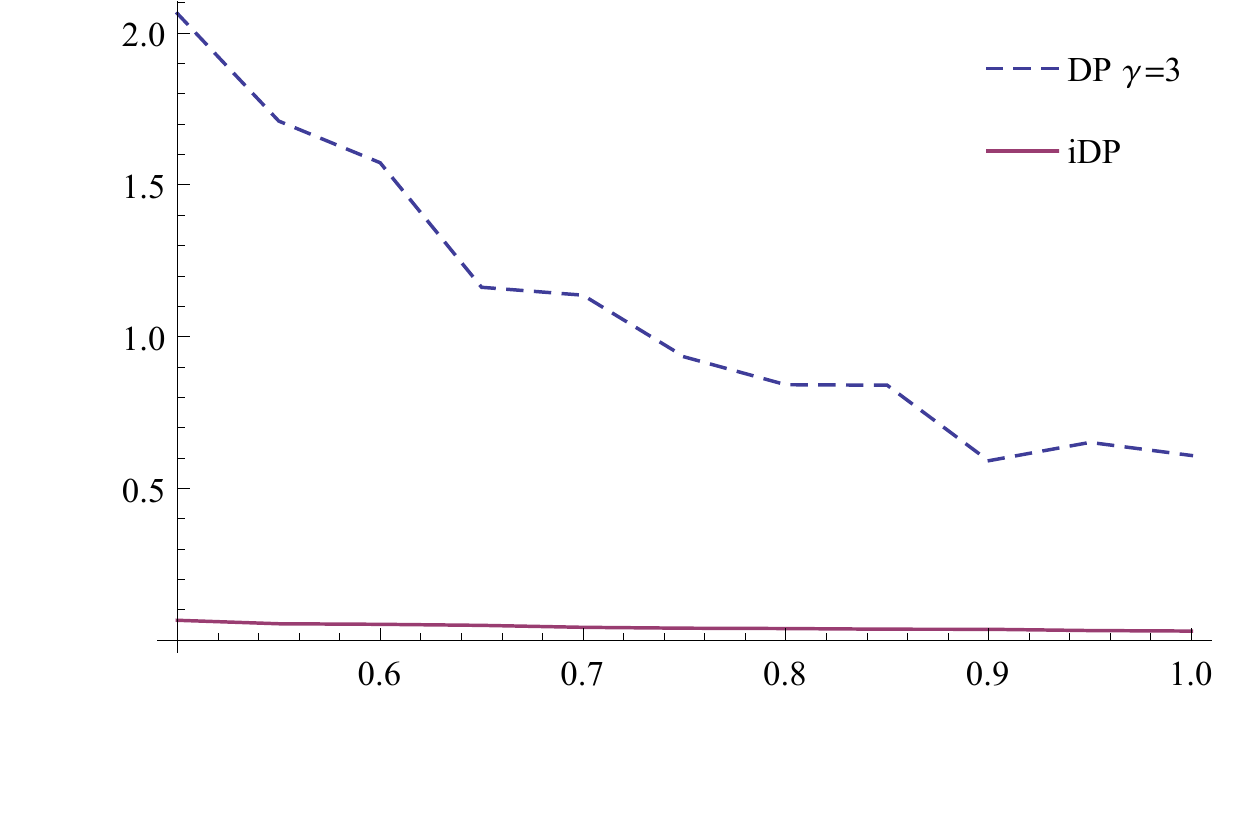}
		\includegraphics[width=3.5cm]{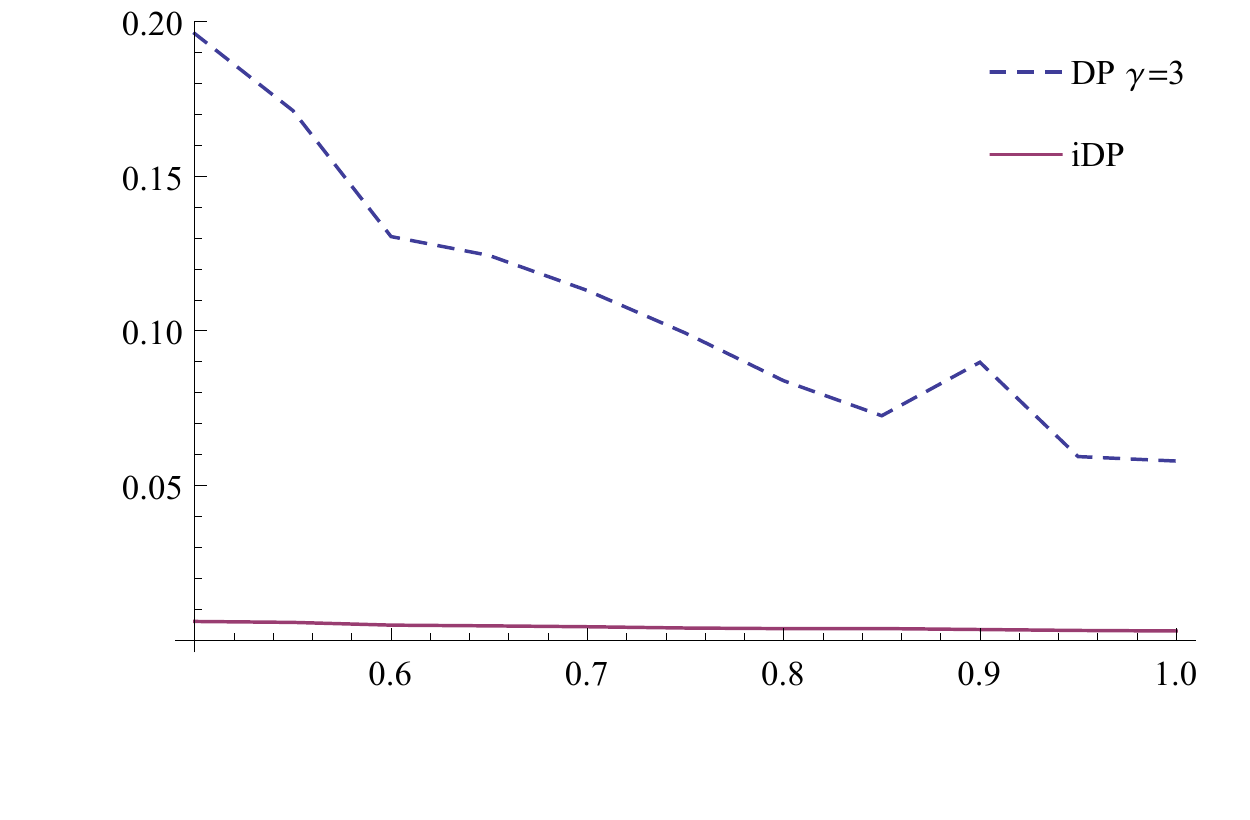}
		\par\end{centering}
	
	\caption{\label{fig:comp}Absolute error in the median (y-axis) for 
$\epsilon$ in the range $[0.5,1]$ 
		(x-axis) and several data sets. Data sets have been drawn from
		a specific distribution: $U[0,1]$ (top row graphs), 
$N(0,1)$ (middle row graphs) 
		and $\mathit{Exp}(1)$ (bottom row graphs). 
Data set sizes are 10 (left column graphs), 100 (middle column graphs), 
		and 1000 (right column graphs).
Especially in the left column graphs, it can be seen that
strict DP yields completely off-range median values for all
considered $\epsilon$,
as the median of a $U[0,1]$ must be within $[0,1]$, and, with very great
probability, the
median of a $N(0,1)$ is within $[-3,3]$ and the median of an $Exp(1)$
is within $[0,5]$. Clearly, iDP is much more utility-preserving.}
\end{figure*}

\subsection{Maximum}

Rather than comparing the accuracy obtained with different mechanisms
in a baseline scenario, this example focuses on the maximum value
of a data set and shows how iDP is useful to avoid calibrating the
noise to a domain-dependent (and hence too large) sensitivity.

\subsubsection{Differential privacy via calibration to the global sensitivity}

The global sensitivity of the maximum of a set of values equals the
length of the domain. Let us consider data sets $D_{1}$ and $D_{2}$
such that: in $D_{1}$, all the values $x_{1},\ldots,x_{n}$ are equal
to the minimum of $Dom(A)$, and in $D_{2}$ the values $x_{1},\ldots,x_{n-1}$
are equal to the minimum of $Dom(A)$, and the value $x_{n}$ is equal
to the maximum of $Dom(A)$. Thus, the global sensitivity for the
maximum is 
\[
\Delta(maximum)=\max(Dom(A))-\min(Dom(A)).
\]
Like in the case of the median, having a global sensitivity equal
to the domain size severely degrades the accuracy of differentially
private estimations of the maximum via calibration to the global sensitivity.
Also, if the domain is unbounded, the global sensitivity cannot even
be computed.

\subsubsection{Differential privacy via calibration to the smooth sensitivity}

The smooth sensitivity for the maximum can be computed as 
\[
\begin{array}{l}
S_{maximum,\epsilon}(D)=\\
{\displaystyle \max_{k=0,\ldots,n}\{\exp(-k\epsilon)(\max(\textit{Dom}(A))-x_{n-k}),}\\
\exp(-k\epsilon)(x_{n}-x_{n-k-1})\}.
\end{array}
\]
Calibration to the smooth sensitivity is effective to reduce the amount
of noise that is added to most data sets. In fact, except for the
worst-case data sets described when computing the global sensitivity,
the smooth sensitivity is smaller. However, the smooth sensitivity
does not avoid the dependence on the domain of the attribute.

\subsubsection{iDP (via calibration to the local sensitivity)}

The local sensitivity of the maximum can be computed as 
\[
LS_{maximum}(D)=\max\{\max(\textit{Dom}(A))-x_{n},x_{n}-x_{n-1}\}.
\]
While iDP offers the expected advantages (the local sensitivity is
smaller than the smooth sensitivity, and we can use exponentially
decreasing noise distributions), the local sensitivity still depends
on the domain of the attribute, which can be unbounded. However, iDP
allows an easy workaround to solve this issue: rather than querying
for the maximum value, we can query for the second maximum. Unless
there is a single record that has a significantly greater value (as
in the worst-case data set described above), querying for the second
maximum should be a reasonably good approximation; otherwise, the
approximation will be very inexact, but, anyway, the accuracy when
querying for the maximum would also be poor. The advantage of querying
for the second maximum is that the local sensitivity does not depend
on the domain anymore: 
\[
LS_{2-\max}(D)=\max\{x_{n}-x_{n-1},x_{n-1}-x_{n-2}\}.
\]
It is important to note that this workaround to avoid having a sensitivity
that depends on the domain attribute is not possible with calibration
to the smooth sensitivity.

\subsection{Range queries}

Given a data set $D$ and a range $\mathcal{R}$, a range query counts the number
of records of $D$ that are contained in $\mathcal{R}$:
\[
f_{\mathcal{R}}(D)=\left|\{x\in D:\,x\in\mathcal{R}\}\right|.
\]

Range queries behave well under DP: their sensitivity is 1. This low
sensitivity is preserved in histogram queries (a collection of range 
queries over disjoint ranges).

The mechanisms to attain $\epsilon$-iDP described in Sections~\ref{sub:lap}
and~\ref{sub:discrete} were based on calibrating the noise
to the local sensitivity at $D$. As the local sensitivity of $f_{\mathcal{R}}$
at any $D$ is 1, which is equal to the global sensitivity, the use
of these mechanisms does not provide improved accuracy w.r.t. DP. 

\section{Conclusions and future work}

This work formalizes and discusses \emph{individual differential privacy},
an alternative to the standard formulation of DP that reduces the
noise to be added to the query results and, thus, better preserves their
accuracy/utility. While, at first sight, individual differential privacy may
look like another relaxation of DP, it
exactly maintains the intuitive disclosure limitation guarantee of DP: the
presence or absence of one individual in the data set must be unnoticeable
from the query result. Improving the accuracy of the results
is possible because individual differential privacy exploits 
the fact that the actual data set is known by the trusted data controller at
the time of answering queries. By focusing only on indistinguishability
between the actual data set and its neighbor data sets, the sensitivity
of the query and, and hence, the magnitude of noise to be added significantly
decrease.

We have also proposed several mechanisms to attain individual differential
privacy. First, we have explained that any mechanism providing $\epsilon$-DP
also provides $\epsilon$-individual differential privacy. However,
direct use of $\epsilon$-differentially private mechanisms fails
to reap the potential accuracy improvements of individual differential
privacy. Next, we have shown that, for numerical queries, $\epsilon$-iDP
can be attained by adjusting the noise to the local sensitivity, which
results in substantial accuracy gains: 
\begin{itemize}
\item Since the local sensitivity is normally significantly smaller than
the global sensitivity employed by standard DP, iDP leads to substantially
better accuracy. 
\item Even if noise is calibrated to the smooth sensitivity, iDP still offers
much better accuracy and this for two reasons: (i) the local sensitivity
is smaller than the smooth sensitivity (although the difference is
not as important as with respect to the global sensitivity); (ii) exponentially
decreasing random noise can be used with local sensitivity (in contrast
with the heavy-tailed noises that are to be employed 
with calibration to the smooth sensitivity). 
\end{itemize}
In addition to improved accuracy, iDP via local sensitivity is less
dependent on the attribute domain (which may be large or even unbounded);
what is more, in case the local sensitivity 
depends on the domain, workarounds can be
found.

In our opinion, the significant accuracy/utility gains brought by
$\epsilon$-iDP, together with its strong privacy guarantee (the same
intuitive privacy guarantee of $\epsilon$-DP for individuals), pave
the way to using $\epsilon$-iDP where standard DP is not viable.
As future work, we plan to study the performance of iDP for other
common queries in the interactive scenario.
We also plan to design mechanisms, other than noise addition, that may
offer improved utility for specific tasks. Specifically,
we aim at mechanisms for non-numerical discrete functions
that leverage the advantages of iDP.
Finally, we also plan to apply
$\epsilon$-iDP to non-interactive data releases. Although such releases
provide more flexibility regarding data uses, they have been traditionally
neglected by researchers because of the enormous distortion that making
them differentially private according to the standard definition would
entail~\cite{Soria2014}.

\section*{Acknowledgments and disclaimer}

We thank Frank McSherry and Daniel Kifer for their
comments and suggestions on individual differential privacy.
This work was supported by the European Commission (projects H2020-644024
``CLARUS'' and H2020-700540 ``CANVAS''), by the Spanish Government
(projects TIN2014-57364-C2-1/2-R ``SmartGlacis'', 
TIN2016-80250-R ``Sec-MCloud'' and TIN2015-70054-REDC)
and by the Government of Catalonia under grant 2014 SGR 537. Josep
Domingo-Ferrer is partially supported as an ICREA-Acadèmia researcher
by the Government of Catalonia. The opinions expressed in this paper
are the authors' own and do not necessarily reflect the views of UNESCO.
The first and second authors thank the Isaac Newton Institute for Mathematical
Sciences at the University of Cambridge  
for support and hospitality during the 
``Data Linkage and Anonymisation'' programme 
(funded by EPSRC Grant Number EP/K032208/1),
when work on this paper was completed and discussed.

\providecommand{\url}[1]{#1} \csname url@samestyle\endcsname \providecommand{\newblock}{\relax}
\providecommand{\bibinfo}[2]{#2} \providecommand{\BIBentrySTDinterwordspacing}{\spaceskip=0pt\relax}
\providecommand{\BIBentryALTinterwordstretchfactor}{4} \providecommand{\BIBentryALTinterwordspacing}{\spaceskip=\fontdimen2\font plus
\BIBentryALTinterwordstretchfactor\fontdimen3\font minus
  \fontdimen4\font\relax} \providecommand{\BIBforeignlanguage}[2]{{%
\expandafter\ifx\csname l@#1\endcsname\relax
\typeout{** WARNING: IEEEtran.bst: No hyphenation pattern has been}%
\typeout{** loaded for the language `#1'. Using the pattern for}%
\typeout{** the default language instead.}%
\else
\language=\csname l@#1\endcsname
\fi
#2}} \providecommand{\BIBdecl}{\relax} \BIBdecl

\end{document}